\documentclass[a4paper, 12pt]{article} 
\usepackage[german,english]{babel} 
\usepackage{booktabs} 
\usepackage{comment} 
\usepackage[utf8]{inputenc} 
\usepackage[T1]{fontenc} 

\usepackage{amsmath,amsfonts,amsthm} 
\usepackage{tikz} 
\usepackage{tikz-cd}
\usetikzlibrary{positioning,arrows} 
\usetikzlibrary{decorations.pathreplacing} 
\usepackage{tikzsymbols} 
\usepackage{blindtext} 

\usepackage{geometry} 

\usepackage{setspace} 
\usepackage[T1]{fontenc} 
\usepackage[utf8]{inputenc} 

\usepackage{tikz}
\usetikzlibrary{trees}
\usepackage{amsmath,amscd,amsbsy,amssymb,latexsym,url,bm,amsthm}
\usepackage{epsfig,graphicx}
\usepackage{enumitem,balance}
\usepackage{wrapfig}
\usepackage{mathrsfs,euscript}
\usepackage{hyperref}

\usepackage{float}
\usepackage{geometry}
\usepackage{listings}
\usepackage{physics}

\usepackage{amssymb}

\usepackage{algorithm} 
\usepackage{algpseudocode}
\usepackage{amsmath,amscd,amsbsy,amssymb,latexsym,url,bm}
\usepackage{epsfig,graphicx}
\usepackage{balance}
\usepackage{wrapfig}
\usepackage{mathrsfs,euscript}
\usepackage{hyperref}

\usepackage{float}
\usepackage{geometry}
\usepackage{listings}
\usepackage{physics}

\usepackage{algorithm} 
\usepackage{algpseudocode} 
\usepackage[group-separator={,}]{siunitx}
\usepackage{subcaption}

\usepackage{enumitem}

\tikzset{
treenode/.style = {circle, minimum size=#1,inner sep=2pt, outer sep=0pt},
treenode/.default = 25pt 
}

\tikzset{
rectnode/.style = {rectangle,minimum size=#1,inner sep=4pt, outer sep=0pt},
rectnode/.default = 20pt 
}

\usetikzlibrary{matrix}
\tikzset{r/.style={fill=red}}
\tikzset{b/.style={fill=cyan}}
\tikzset{w/.style={fill=white}}

\tikzset{boardstyle/.style={matrix of nodes,
        nodes in empty cells,
        row sep=-\pgflinewidth,
        column sep=-\pgflinewidth,
        nodes={draw,minimum width=0.3cm,minimum height=0.3cm,anchor=center}}}

\providecommand{\keywords}[1]{\textbf{\textit{Keywords:}} #1}

\newtheorem{theorem}{Theorem}
\newtheorem{corollary}{Corollary}

\newtheorem{example}{Example}
\newtheorem{proposition}{Proposition}
\newtheorem{definition}{Definition}

\usepackage[sc]{mathpazo} 

\title{Constrained Optimal Querying:\\ Huffman Coding and Beyond} 
\author{
Shuyuan Zhang\and Jichen Sun\and Shengkang Chen \\
 }

\date{\small May 2022} 

\begin{document}


\maketitle 


\begin{abstract}

Huffman coding is well known to be useful in certain decision problems involving minimizing the average number of (freely chosen) queries to determine an unknown random variable. However, in problems where the queries are more constrained, the original Huffman coding no longer works. In this paper, we proposed a general model to describe such problems and two code schemes: one is Huffman-based, and the other called GBSC (Greedy Binary Separation Coding). We proved the optimality of GBSC by induction on a binary decision tree, telling us that GBSC is at least as good as Shannon coding. We then compared the two algorithms based on these two codes, by testing them with two problems: DNA detection and 1-player Battleship, and found both to be decent approximating algorithms, with Huffman-based algorithm giving an expected length $1.1$ times the true optimal in DNA detection problem, and GBSC yielding an average number of queries $1.4$ times the theoretical optimal in 1-player Battleship.

\end{abstract}



\setcounter{page}{1} 

\section{Introduction} 

Our research is inspired by a problem in Cover's textbook\cite{cover1999elements} (see page 153).

\begin{example}[Bad wine]
\label{badWine}
One is given six bottles of wine. It is known that precisely one bottle has gone bad (tastes terrible). From inspection of the bottles it is determined that the probability $p_i$ that the $i$th bottle is bad is given by $(p_1, p_2, \dots, p_6) = (\frac{8}{23},\frac{6}{23},\frac{4}{23},\frac{2}{23},\frac{2}{23}, \frac{1}{23})$. Tasting will determine the bad wine. You can mix some of the wines in a fresh glass and sample the mixture. You proceed, mixing and tasting, stopping when the bad bottle has been determined. What is the minimum expected number of tastings required to determine the bad wine?
\end{example}

As an exercise in the textbook, this problem is well solved. In fact, it is equivalent to Huffman coding. The process that we determine which bottle of wine is bad can be regarded as a decision tree. At each node, we make an observation that has two possible outcomes (in this example, good or bad) and move to the left or right child node according to the outcome. This process continues until we reach a leaf node. If we represent each move to left with $0$ and each move to right with $1$, then a move sequence can be represented by a binary code. Minimizing the expected number of moves is equivalent to minimizing the expected code length, so we only need to use Huffman coding.

We want to generalize this problem. A natural generalization is changing the number of bad wines. It seems that the original Huffman coding strategy will still work, but unfortunately, it is wrong. This is shown by the following example.

\begin{example}[More bad wine]
\label{moreBadWine}
The background is similar to Example \ref{badWine}, but this time there are four bottles of wine to be examined and two of them are bad. Suppose the probability $p_{ij}$ that the $i$th and $j$th bottles are bad is given by $(p_{12},p_{13},p_{14},p_{23},p_{24},p_{34})=(0.1,0.1,0.15,0.15,0.3,0.2)$. What is the minimum expected number of tastings required to determine the two bad wines?
\end{example}

If we repeat the Huffman coding strategy, we will obtain the unique Huffman tree below.

\begin{figure}[H]
    \centering
    \begin{tikzpicture}[level distance=1.5cm,
      level 1/.style={sibling distance=4cm},
      level 2/.style={sibling distance=2cm}]
      \node [treenode, draw] {$1$}
        child
        {
            node [treenode, draw] {$0.4$}
            child 
            {
                node [treenode, draw]{$0.2$}
                child 
                {
                    node [treenode, draw]{$p_{12}$}
                }
                child 
                {
                    node [treenode, draw]{$p_{13}$}
                }
            }
            child 
            {
                node [treenode, draw]{$p_{34}$}
            }
        }
        child
        {
            node [treenode, draw] {$0.6$}
            child 
            {
                node [treenode, draw]{$p_{24}$}
            }
            child 
            {
                node [treenode, draw]{$0.3$}
                child 
                {
                    node [treenode, draw]{$p_{14}$}
                }
                child 
                {
                    node [treenode, draw]{$p_{23}$}
                }
            }
        };
    \end{tikzpicture}
    \caption{Huffman Tree of Example \ref{moreBadWine}.}
\end{figure}
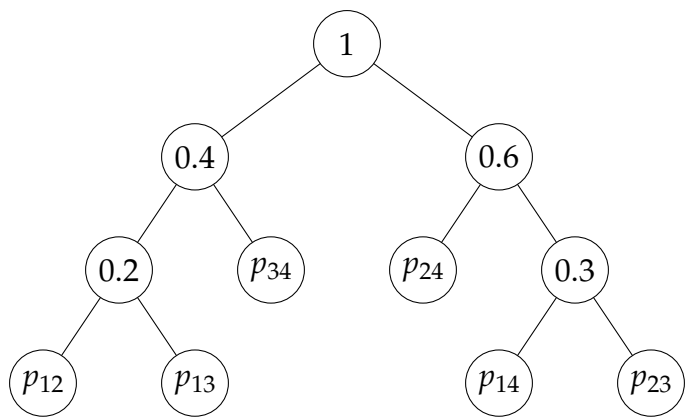

The expected Huffman code length is $2.5$, but it is not the answer. This is because the decision tree we have constructed is infeasible. To see this, just look at the root of the tree. Note that in both the left and right subtrees, every bottle of wine may be bad, so there is no way that we can decide in which direction we should move by just one observation.



Example \ref{moreBadWine} shows that changing Example \ref{badWine} a little bit will produce a much more difficult problem. In fact, Example \ref{moreBadWine} is just an instance of a huge family of similar problems, and finding the precise solution of these problems is usually extremely difficult.

In this report, we proposed a generalized model that formulates such type of problems, and analyzed three problems related to this model. We designed two generalized approximate algorithms for these problems. One algorithm is based on Huffman coding, while the other is based on a new code scheme (at least not being researched much). We call the new code GBSC (Greedy Binary Separation Code), and we proved an interesting property of GBSC.
\section{Models}
\subsection{Generalized Model}

We consider this generalized problem. Given a random variable $X$ with alphabet $\mathcal{X}$ and its distribution $p(x)$. We determine the value of $X$ by asking questions. The $i$-th question we ask is that "is $X$ in set $S_i$?", where $S_i \in \mathscr{A}, \forall i$. Suppose we need $N$ questions to determine the value of $X$. Our goal is to find out the minimum of $EN$ (the expectation of $N$) and if possible, its corresponding strategy.

\begin{definition}

The family of set $\mathscr{A}$ is called a decision set on $\mathcal{X}$.

\end{definition}

Obviously, different structures of the decision set will produce a bunch of completely different problems. In the rest of this chapter, we will introduce three specializations of this model, which are Huffman coding, the DNA detection problem and the Battleship problem.

\subsection{Huffman Coding: the Easy Case}

In Example \ref{badWine}, the decision set is $2^\mathcal{X}$. We have seen that in this case, the problem is equivalent to Huffman coding. However, the condition can be loosened slightly.

\begin{definition}

A decision set $\mathscr{A}$ on $\mathcal{X}$ is  decision-complete, if $\forall S\in 2^\mathcal{X} \setminus \{\emptyset, \mathcal{X}\}$, $S \in \mathscr{A} $ or $\mathcal{X} \setminus S \in \mathscr{A}$.

\end{definition}

\begin{proposition}
\label{decisionCompleteFeasible}
If $\mathscr{A}$ is decision-complete, then any decision tree of $X$ is feasible.
\end{proposition}

\begin{proof}
This is easy to demonstrate. Asking "is $X$ in set $S_i$?" is equivalent to asking "is $X$ in set $\mathcal{X} \setminus S_i$?", and there is no point to ask whether $X$ is in $\emptyset$ or $\mathcal{X}$.
\end{proof}

\begin{corollary}
\label{equivHuffman}
If $\mathscr{A}$ is decision-complete, then Huffman coding will produce the optimal solution.
\end{corollary}


Proposition \ref{decisionCompleteFeasible} also provides a trivial but practical necessary condition for a set to be decision-complete.

\begin{corollary}
\label{minNumA}
If $\mathscr{A}$ is decision-complete, then $|\mathscr{A}| \ge 2^{|\mathcal{X}|-1} - 1$.
\end{corollary}

Using Corollary \ref{minNumA}, it is easy to prove that the decision set of Example \ref{moreBadWine} is not decision-complete,  so the Huffman coding strategy might fail.

\begin{proposition}
The decision set $\mathscr{A}$ of Example \ref{moreBadWine} is not decision-complete.
\end{proposition}

\begin{proof}
$|\mathcal{X}|=6$, so $|\mathscr{A}| = 2^4 - 1 < 2 ^{ |\mathcal{X}|-1} - 1$.
\end{proof}

\subsection{DNA Detection Problem}
To determine genomic sequences of several organisms, biological meaning needs to be assigned to particular regions of the sequence. One of important steps in this process is the identification of genes. An \textit{Exon} is an interval of the DNA sequence and it does not overlap with other exons and gene is a sequence of exons. In this paper, we simplified the question by assuming that the target gene only contains one exon. We can detect whether the target exon is in an interval in the DNA sequence or not at each detection. The position of the exon on DNA is fixed, so we want to minimize the expected number of detection to determine the position of the exon by choosing the intervals wisely, thus reducing the cost of DNA detection\cite{biedl2004finding}\cite{xu1998gene}.

\begin{definition}
A set $S$ of integers is continuous, if $S = [\min S, \max S] \cap \mathbb{Z}$.
\end{definition}

\begin{example}
$\{1,2,3\}$ and $\{4\}$ are continuous, while $\{1,3\}$ is not.
\end{example}

If we assume that the exon's location we want to find on the DNA is discrete and unique, we can use a random variable $X\in \{1,2,\dots, n\}$ to it, and in this case, the decision set is any continuous set in space $\mathcal{X}$.

\begin{figure}
    \centering
    \includegraphics[width=0.8\textwidth]{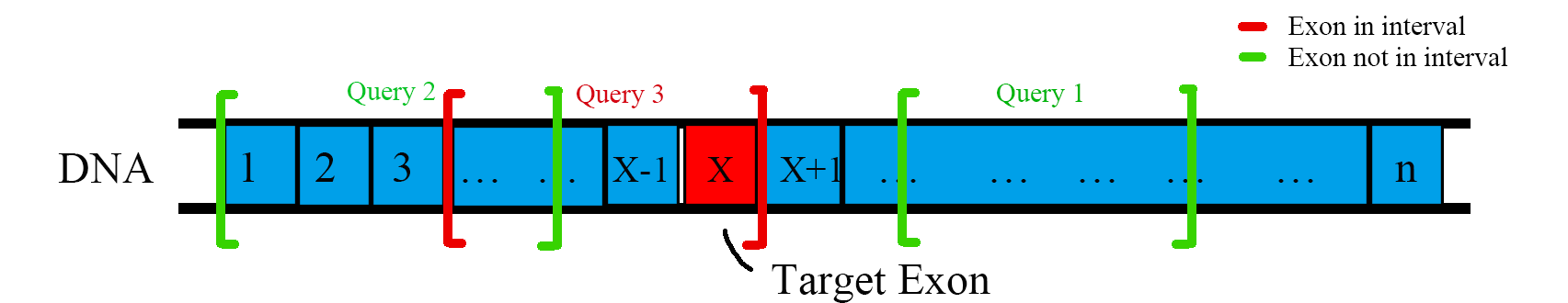}
    \caption{DNA exon detection.}
    \label{fig:exon}
\end{figure}

\subsection{Battleship Problem}
"Battleship" is a popular 2-player strategy and guessing game. In the typical setting, each player places "boats" with different lengths on his $10\times 10$ board, which is hidden to the opponent. Each player take turns to "bomb" a grid $(i,j)$ on opponent's board, and the opponent must honestly report if it "hits" or "misses". The goal of the game is sink all of the opponent's ships, i.e. "hit" all the grids on the opponent's board that represents a ship, before the opponent sinks all the player's ships. 
To simplify analysis and computation, we instead study the 1-player Battleship. In this setting, the game randomly generates a possible layout of ships unknown to the player. The goal of the player is to sink all the ships in the fewest tries (bombs). 
For example, in this particular board placed by the opponent (the game), the player needs to "hit" all grids marked in gray.
\begin{figure}[H]\label{battleshipgame}
    \centering
    \includegraphics[scale=0.65]{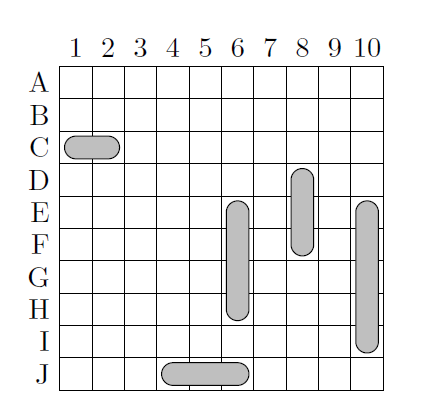}
    \caption{One possible board of one player in Battleship
    \cite{battleshipoptimal}.}
\end{figure}

The 1-player Battleship problem can be formulated into the following. Suppose there are $n$ different possible ship layouts. Let $\mathcal{X}_0 = \{ X_1, X_2, X_3, \dots, X_n \}$ be the set of all possible layouts, where $X_k, 1\leq k \leq n$ is a $10\times 10$, $0-1$ matrix. The game randomly chooses the target board $X^{*} \in \mathcal{X}_0$ (every legal board is equally possible), and the player tries to minimize the number of tries $T$ to guess $T^{*}$.

At $t$-th step, the player gives a query on i, j: $Q_t  = (i_t,j_t)$, which means: "Is bombing of grid $(i,j)$ a hit?" the game answers honestly, giving the player some information to eliminate some layouts in $\mathcal{X}_{t-1}$ and get the new $\mathcal{X}_{t}$.
$$\mathcal{X}_t = \{X|(X\in \mathcal{X}_{t-1}) \wedge (X_{ij} = X^{*}_{ij})\}$$
The goal is to minimize the number of tries $T$ to determine what the target layout is,  by choosing queries wisely. That is, 

$$    \min_{Q_1, Q_2, \dots, Q_T} T  , Q_t = (i_t, j_t)$$
$$    \mathcal{X}_t = \{X|(X\in \mathcal{X}_{t-1}) \wedge (X_{i_t j_t} = X^{*}_{i_t j_t})\} , 1 \leq t \leq T$$
$$\mathcal{X}_0 = \{ X_1, X_2, X_3, \dots, X_n \},	\mathcal{X}_T = \{X^{*}\}, \text{ random variable } X^* \in \mathcal{X}_0 $$

More precisely, we would like to minimize the average number of tries $\bar{T}$ over all possible targets $X^{*}$.

Previous studies use diagonal searching\cite{Rodin}, tree search\cite{battleshipoptimal}, or RL-based\cite{battleshiprl} approaches. In our study, we adopt the idea of greedy information gain per step, which will be demonstrated in Chapter \ref{sec:gbsc}.

\section{Greedy Algorithm Based on Huffman coding}

\subsection{Description}

	

The most straightforward algorithm is search by brute force. This algorithm is promised to find the precise result but is too slow for large-scale problems. Therefore, we try to use the same method of Huffman coding. In every iteration we still try to merge the two nodes with minimum sum of probabilities, if possible. This gives us the prototype of a greedy algorithm based on Huffman coding.

\begin{algorithm}
	\caption{Greedy algorithm based on Huffman coding} 
	\begin{algorithmic}[1]
	\Procedure{GreedyHuffman}{nodes}
	    \If{$nodes$ contains only one node}
	        \State \Return $nodes$
	    \EndIf
	    \State $pairs \gets \{(nodes[i], nodes[j]): i < j\}$
	    \State sort $pairs$ by the sum of probabilities in ascending order
        \For{every $(i,j) \in pairs$}
        		\If{$i$ and $j$ can merge}
        		    \State $newNodes \gets$ merge $i$ and $j$
        		    \State $tree \gets$ \Call{GreedyHuffman}{$newNodes$}
        		    \If{$tree \ne \emptyset$}
        		        \State \Return $tree$
        		    \EndIf
                \EndIf
        \EndFor
        \State \Return $\emptyset$
    \EndProcedure
	\State read $p(x)$ and initialize corresponding $nodes$
	\State $tree \gets$ \Call{GreedyHuffman}{$nodes$}
	\end{algorithmic} 
	
\end{algorithm}

The critical part of the algorithm is how we judge whether two nodes can be merged (line 8). The following proposition describes this process mathematically, but it is not practical in practice. The algorithm of the process needs to be designed specifically to get a better performance.

\begin{definition}
The set of possible values of $X$ at a certain node of a decision tree is called the candidates of the node.
\end{definition}

\begin{proposition}
Assume that $A$ and $B$ are the candidates of two nodes. Then the two nodes can be merged if and only if $\exists C \in \mathscr{A}$, $(A \cup B)\setminus C \in \{A, B\}$.
\end{proposition}

\subsection{First Try to Solving the DNA Detection Problem}

The following proposition instructs us how to judge whether two nodes can be merged for this problem.

\begin{proposition}
Assume that $A$ and $B$ are the candidates of two nodes. Then the two nodes can be merged if and only if any of the two conditions holds:
    \begin{enumerate}[label=(\arabic*)]
        \item $A$ or $B$ is continuous.
        \item $\min A > \max B$ or $\min B > \max A$.
    \end{enumerate}
\end{proposition}

Using these two conditions, we implemented our first algorithm that deals with the DNA detection problem. We carried out a experiment by generating \num{10000} random $p(x)$ with $n=6$ and compare the result of the brute force algorithm and the Huffman-based algorithm. The results of the experiment indicate that the result of this algorithm is rather close to the optimal value. The result of each data is represented by a point in Figure \ref{fig:greedy_huffman_1}. The red line in Figure \ref{fig:greedy_huffman_1} is the optimal bound, because the brute force algorithm is promised to find the precise optimal value. We can see that most of the points are close to the red line, and many of them exactly lie on the red line, which means they reach the optimal bound. In fact, around $60\%$ of the data reach the optimal bound.

This intuition is further confirmed by Figure \ref{fig:greedy_huffman_2}. Let $L_b$ and $L_g$ be the expected length of the brute force algorithm and the Huffman-based algorithm, respectively. We define that $\text{gap} = \frac{L_g-L_b}{L_b}$. Although the maximum of gap is around $35\%$, in most cases, the gap is less than $10\%$. Therefore, the result of the Huffman-based algorithm is very satisfactory.

\begin{figure}[H]
    
    \centering
    
    \begin{subfigure}{0.45\textwidth}
        \includegraphics[width=1.0\textwidth]{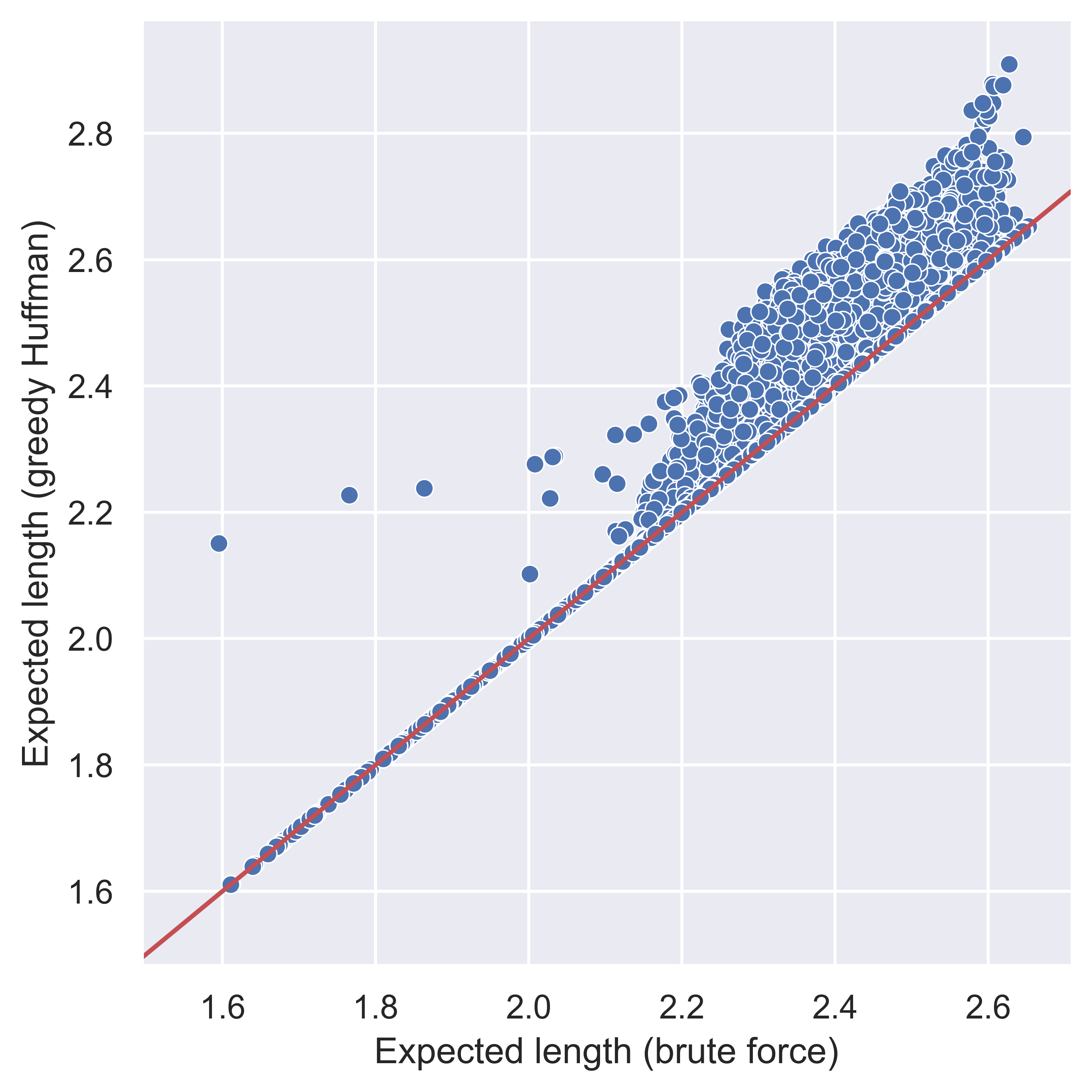}
        \caption{Expected lengths of the brute force algorithm and the Huffman-based algorithm.}
    \label{fig:greedy_huffman_1}
    \end{subfigure}
    \begin{subfigure}{0.45\textwidth}
    \includegraphics[width=1.2\textwidth]{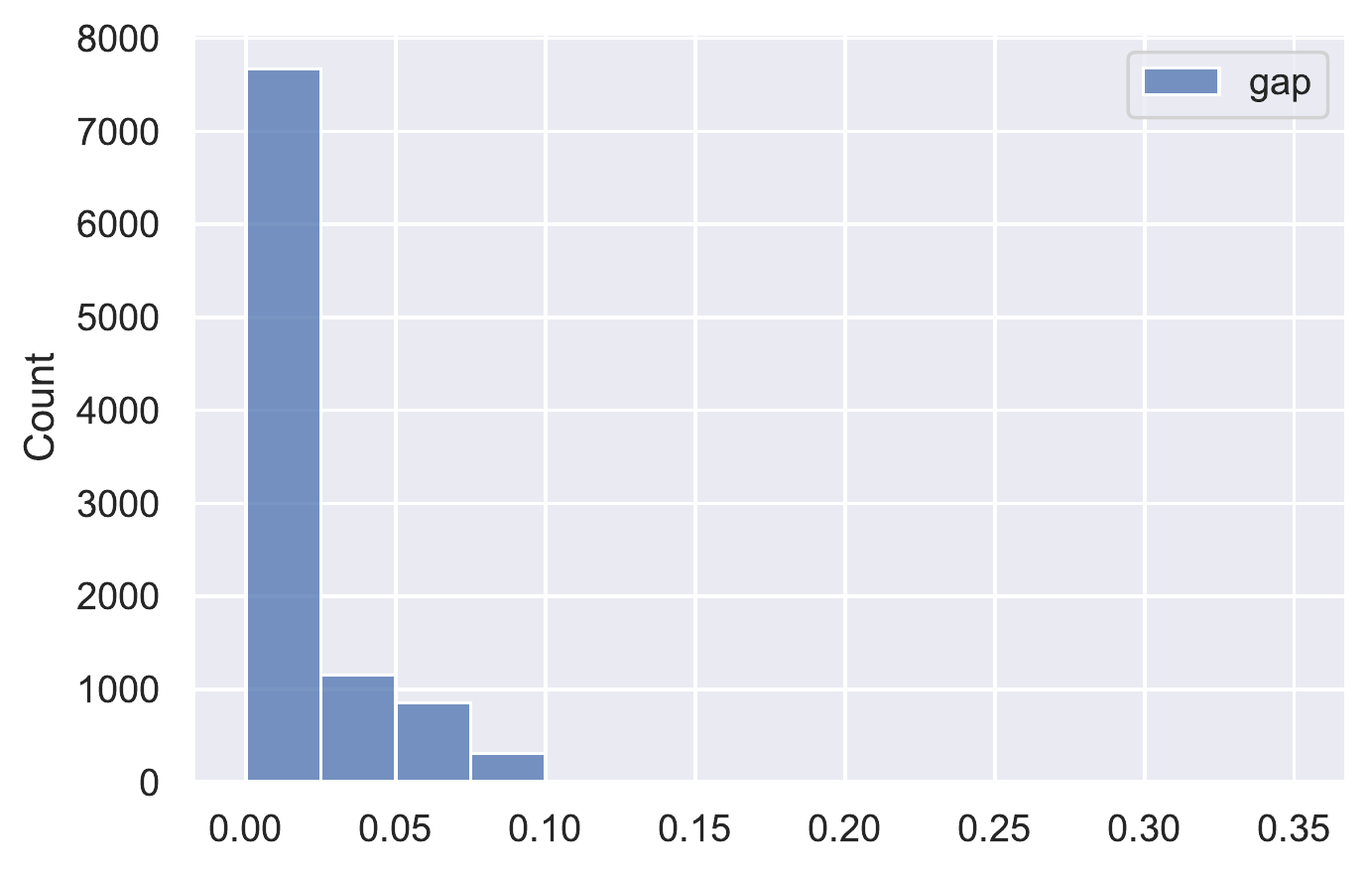}
    \caption{Distribution of the gap between the optimal value and the output of the Huffman-based algorithm.}
    \label{fig:greedy_huffman_2}
      \end{subfigure}
    \caption{Comparison between Huffman-based algorithm and the brute force optimal}
\end{figure}

\section{Greedy Binary Separation Coding}
\label{sec:gbsc}

\subsection{Motivation}

The greedy algorithm based on Huffman coding has a vital drawback: when $|\mathcal{X}|$ is large and $|\mathscr{A}|$ is relatively small, the algorithm becomes very slow. In most cases, $|\mathscr{A}|$ is small compared to $|\mathcal{X}|$. Therefore, instead of building the decision tree from bottom to top, we consider another way of building the tree from top to bottom, by choosing the best question. We call this coding Greedy Binary Separation Coding (GBSC).

\subsection{Definition}

We now formally describe GBSC.

\begin{definition}
$\{A, B\}$ is a (binary) partition of $S$, if $A \cup B = S$ and $A \cap B = \emptyset$.
\end{definition}

\begin{definition}
A partition $\{A,B\}$ of $S$ is optimal, if for any partition $\{C,D\}$ of $S$, $|p(A) - p(B)| \le |p(C) - p(D)|$, where $p(X)$ is the sum of all the numbers in $X$.
\end{definition}

The idea of GBSC is simple. We construct the decision tree from top to bottom. Let $S=\{p(x): x \in \mathcal{X} \}$. At root, we choose any optimal partition of $S$, and split the tree according to the partition. The process is repeated recursively at each node, and each time we try to find the best partition of the candidates of the node.

\begin{example}
\label{gbscExample}
If $X\in\{1,2,3,4\}$ and $(p_1,p_2,p_3,p_4)=(0.1,0.2,0.3,0.4)$, then the decision tree of $X$ using GBSC is as follows.

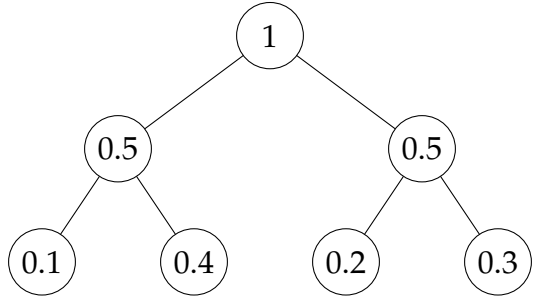
\begin{figure}[H]
    \centering
    \begin{tikzpicture}[level distance=1.5cm,
      level 1/.style={sibling distance=4cm},
      level 2/.style={sibling distance=2cm}]
      \node [treenode, draw] {$1$}
        child
        {
            node [treenode, draw] {$0.5$}
            child 
            {
                node [treenode, draw] {$0.1$}
            }
            child 
            {
                node [treenode, draw]{$0.4$}
            }
        }
        child
        {
            node [treenode, draw] {$0.5$}
            child 
            {
                node [treenode, draw]{$0.2$}
            }
            child 
            {
                node [treenode, draw]{$0.3$}
            }
        };
    \end{tikzpicture}
    \caption{Decision tree of Example \ref{gbscExample}.}
\end{figure}

\end{example}

\subsection{Intuition}
The intuition of GBSC is simple: greedily maximize the information gain (mutual information) at each query. In binary decision trees, each node corresponds to a possible random variable sampled from some distribution. For a query $Q$ on node $X\in \mathcal{X} = \{X_1,X_2,\dots, X_n\}$, each $X_i$ gives a binary answer.
The query is designed such that 
$$
p = \frac{| \{X_i|X_i\in \mathcal{X}  \wedge \text{ query on } X_i \text{ gives result 1} \}|} {n}
$$
so $\Pr(Q=1) = p$. 
The information gain of a query $Q$ with respect to node $X$ is 
$$
\begin{aligned}
I(X;Q) &= H(X) - H(X|Q) \\
&= H(X) - \Pr(Q=0) H(X|Q=0) - \Pr(Q=1)H(X|Q=1)
\end{aligned}
$$
Suppose $X$ follows a uniform distribution, then $X$ is also conditionally uniformly distributed given $Q$. Then 
$$
\begin{aligned}
I(X;Q)&=\log n- (1-p)\log n(1-p) - p\log np \\
&= -(1-p)\log (1-p) - p\log p \\
&= H(p)
\end{aligned}
$$

Therefore, $p=\frac{1}{2}$ maximizes the information gain at each node. For more general distributions, we will prove its optimality by the analysis below.

\subsection{Analysis}

Recall that Huffman coding is the best we can do. We hope that GBSC can be as good as Huffman coding. However, sometimes GBSC cannot reach the optimal bound. For $X$ in Example \ref{gbscExample}, the expected code length is $1.9$ for Huffman coding and $2$ for GBSC. Although GBSC is not the best coding, from experiments we observe that most of the time it is quite good, so we believe that the average code length of GBSC will not be too far away from $H(X)$. This inspired us to prove the following theorem (see \textbf{Proposition \ref{GBSCProof}} for proof).

\begin{theorem}
\label{gbscShannon}
Assume that the expected code length is $L_g$ for GBSC and $L_s$ for Shannon coding. Then $L_g \le L_s$, i.e., GBSC is at least as good as Shannon coding.
\end{theorem}

Recall that $L_s$ also satisfies $L_s < H(X)+1$, so we immediately obtain a good upper bound for $L_g$.

\begin{corollary}
Assume that the expected code length is $L_g$ for GBSC. Then $L_g < H(X) + 1$.
\end{corollary}

This gives us confidence that GBSC can be applied to obtain a quite good approximation of the optimal coding. Also, it shows that if we use the same technique mentioned in Cover's textbook (see page 114), we can use GBSC to approach the Shannon bound.

Now we focus on proving Theorem \ref{gbscShannon}. Recall that by definition, the code length for a symbol with probability $p$ in Shannon coding is $\lceil -\log p \rceil$, so we only need to prove this proposition.

\begin{proposition} \label{GBSCProof}
\label{prop1}
Assume that the code length for a symbol with probability $p$ is $L$ for GBSC. Then $L \le n$ if $p \ge 2^{-n}(n \in \mathbb{Z^+})$ , or equivalently, $L \le \lceil -\log p \rceil$.
\end{proposition}

\begin{proof}




In Figure \ref{fig:gbscProof}, a circle represents a single node and a rectangle  represents a leaf or a subtree. The content in the nodes are indices of the nodes. Node $i(1)$ and node $i(2)$ are in the $i$-th layer. $p_{i(j)}$ is the sum of probabilities of all the children nodes of node $i(j)$. 


Basically, we prove by contradiction. WLOG, assume that $p_{k+1(2)} \ge 2^{-k}$. Then we want to prove that the following propositions hold for any integer $n \in [1, k]$ by induction.
\begin{enumerate}[label=(\roman*)]
\item $p_{n(2)} \ge 2^{-n}$;
\item $p_{n-1(1)}\ge 2^{-(n-1)} + p_{k+1(1)}$.
\end{enumerate}
If (i) and (ii) holds, then $p_{0(1)} \ge 1 + p_{k+1(2)} > 1$, causing a contradiction.

\textbf{Basic step}. We first prove the case when $n=k$. If $p_{k+1(1)} < p_{k(1)} - p_{k(2)}$, then moving node $k+1(1)$ to node $k(2)$ will make $|p_{k(1)} - p_{k(2)}|$ smaller\footnote{We do not want to describe the process of "moving" mathematically, because it will just make things harder to understand. In case that some readers may feel confused, we explain the idea more clearly. Simply speaking, moving node $i$ to node $j$ means moving all the leaves in node $i$ to node $j$ so that they become the leaves of node $j$. If there are many leaves to be moved, the process of moving is not unique.}, so $p_{k+1(1)} \ge p_{k(1)} - p_{k(2)}$, i.e., $p_{k(2)}\ge p_{k(1)} - p_{k+1(1)} = p_{k+1(2)} \ge 2^{-k}$. (i) is true. $p_{k-1(1)} = p_{k+1(1)} + p_{k+1(2)} + p_{k(2)} \ge 2^{-(k-1)} + p_{k+1(1)}$. (ii) is true.

\textbf{Inductive step}. Assume that the proposition is true for $n=m+1(1\le m \le k-1)$. If $p_{k+1(1)} < p_{m(1)} - p_{m(2)}$, then moving node $k+1(1)$ to node $m(2)$ will make $|p_{m(1)} - p_{m(2)}|$ smaller, so $p_{k+1(1)} \ge p_{m(1)} - p_{m(2)}$. By inductive assumption, $p_{m(1)} \ge 2^{-m} + p_{k+1(1)}$, so $p_{m(2)} \ge p_{m(1)} - p_{k+1(1)} \ge 2^{-m}$. (i) is true. $p_{m-1(1)}= p_{m(1)} + p_{m(2)} \ge 2^{-(m-1)} + p_{k+1(1)}$. (ii) is true. Therefore, the case when $n=m$ also holds.

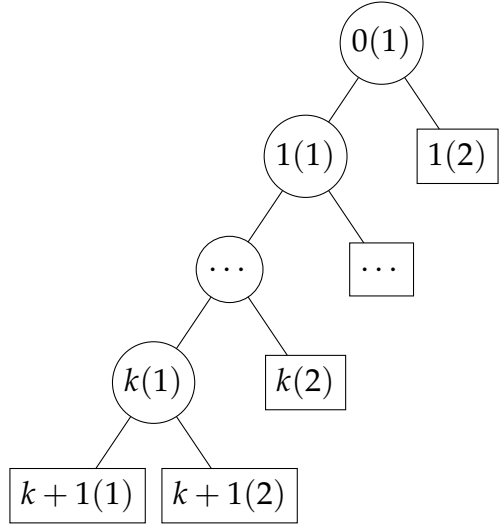
\begin{figure}[H]
    \centering
    \begin{tikzpicture}[level distance=1.5cm,
      level 1/.style={sibling distance=2cm},
      level 2/.style={sibling distance=2cm}]
    \node [treenode, draw] {$0(1)$}
    child 
    {
        node [treenode, draw] {$1(1)$}
        child
        {
            node [treenode, draw]{$\dots$}
            child
            {
                node [treenode, draw]{$k(1)$}
                child
                {
                    node [rectnode, draw]{$k+1(1)$}
                }
                child {node [rectnode, draw]{$k+1(2)$}}
            }
            child {node [rectnode, draw]{$k(2)$}}
        }
        child {node [rectnode, draw]{$\dots$}}
    }
    child { node [rectnode, draw] {$1(2)$} };
    \end{tikzpicture}
    \caption{Illustration for the proof of Proposition \ref{prop1}.}
    \label{fig:gbscProof}
\end{figure}

\end{proof}


\section{Greedy Algorithms Based on GBSC}

\subsection{Basic Structure}

The basic idea of the algorithm is shown below. In practice, the most difficult part is how to find the best question (line 2). However, since usually $|\mathscr{A}|$ is small and $\mathscr{A}$ has a simple structure, this process is often very simple. This is especially the case when we deal with the Battleship problem later.

	

\begin{algorithm}[H]
	\caption{Greedy algorithm based on GBSC} 
	\begin{algorithmic}[1]
	\Procedure{GBSC}{node, p}
	    \State choose question $X \in \mathscr{A}$ that makes the left and right subtrees most average
	    \State $p_1 \gets$ the possibilities in condition that the question is true
	    \State \Call{GBSC}{$root.left, p_1$}
	    \State $p_2 \gets$ the possibilities in condition that the question is false
	    \State \Call{GBSC}{$root.right, p_2$}
    \EndProcedure
	\State read $p$, where $p[i]$ is the probability of $i$
	\State initialize $root$
	\State \Call{GBSC}{$root, p$}
	\end{algorithmic} 
	
\end{algorithm}

\subsection{Dealing with 1-player Battleship Problem Using GBSC}
Recall the 1-player Battleship problem:
$$    \min_{Q_1, Q_2, \dots, Q_T} T  , Q_t = (i_t, j_t)$$
$$    \mathcal{X}_t = \{X|(X\in \mathcal{X}_{t-1}) \wedge (X_{i_t j_t} = X^{*}_{i_t j_t})\} , 1 \leq t \leq T$$
$$\mathcal{X}_0 = \{ X_1, X_2, X_3, \dots, X_n \},	\mathcal{X}_T = \{X^{*}\}, \text{ random variable } X^* \in \mathcal{X}_0 $$

$\mathcal{X}_0$ is a set of $0-1$ matrices representing the entire possible board space without any prior knowledge, and $X^*$ is the target board chosen by the opponent (randomly sampled from $\mathcal{X}_0$ in our case).

More precisely, we would like to minimize the average number of tries $T$ over all possible targets $X$.

By assumption, the target board $X^{*}$ is randomly sampled from $\mathcal{X}_0$. The key to minimizing the number of queries $t$ is to choose the queries $Q_1, Q_2, \dots, Q_T$ such that the size of remaining possible boards, $|\mathcal X_{t}|$ converges to $1$ quickly. 

Using the conclusions from GBSC coding, we can devise a way to minimize the average number of queries.
At $t$-th query, the hit probability matrix is 
$$P^{(t)}_{ij} = \frac{\sum_{X \in \mathcal{X}_{t-1}} X_{ij}}{|\mathcal{X}_{t-1}|}   $$

Obviously, the hit probability of some previously queried $(i,j)$ is either $0$-missed or $1$-hit. So there is no need to query the same grid twice.

we choose $Q_t$ by this greedy strategy, deducted from GBSC:
$$Q^*_t=(i^*_t, j^*_t) = \arg\min_{(i,j)} |P^{(t)}_{ij} - \frac{1}{2}| $$.

Ideally, each query $(i^*_t, j^*_t)$ will half the size of the remaining boards space, giving us the average number of queries $\bar{T} = \lceil \log n \rceil$, where $n = |\mathcal{X}_0|$. Of course, this is not always possible in the real world, since there may not be exist a $(i,j)$ with $p^t_{ij} = \frac{1}{2}$. Nonetheless, we can always choose the grid with hit probability closest to $\frac{1}{2}$ and get a sub-optimal algorithm.

\begin{algorithm}[H]
	\caption{Solving 1-player Battleship using GBSC.} 
	\begin{algorithmic}[1]
    \State calculate initial possible board space $\mathcal{X}_0$
    \State $t \gets 0$
    \While{$|\mathcal{X}_t| > 1$}
        \State $t \gets t + 1$
        \State $P^{(t)}_{ij} \gets \frac{\sum_{X \in \mathcal{X}_{t-1}} X_{ij}}{|\mathcal{X}_{t-1}|}   $
        \Comment{calculate hit probability matrix}
        \State $Q^*_t \gets \arg\min_{(i,j)} |P^t_{ij} - \frac{1}{2}| $ 
        \Comment{choose locally optimal query}
        \State $\mathcal{X}_t \gets \{X|(X\in \mathcal{X}_{t-1}) \land (X_{i^*j^*} = X^{*}_{i^*j^*})\}$
        \Comment{eliminate some boards in $\mathcal{X}_{t-1}$} 
    \EndWhile
    \State $\mathcal{X}_t = \{X^*\}$
	\end{algorithmic} 
	
\end{algorithm}




We did some tests to show the effectiveness of the GBSC-based algorithm on dealing with 1-player Battleship problem. To speed up computation, we consider 3 ships of length 5, 4 and 3 placed on a $10\times 10$ board. In total, there are $n=\num{1850736}$ possible board layouts, which is just under $2^{21}$. Here we define the number of tries $t$ as the number of queries to determine the target board $X^*$.\footnote{Sometimes two or more boards may be different yet indistinguishable, and the terminal entropy is nonzero, for example if two boats of length 3 and 4 are adjacent and from a L-shape pattern. Since this does not affect the outcomes, we consider these boards to be identical.}Since each query (bombing) in average gives at most $1$ bit of information,  the theoretical minimal average number of tries is $\bar{T}^* = \lceil \log n \rceil  = 21$ queries. Of course, one can sometimes get lucky and determine $X^*$ in less than $21$ tries.

Here we have 10 randomly chosen target boards for our algorithm to play against. The vertical axis is calculated by $H(X) = \log |\mathcal{X}_t|$. Again, we assume that every target board has the same probability of being chosen.
\begin{figure}[H]
    \centering
    \includegraphics[width=0.5\textwidth]{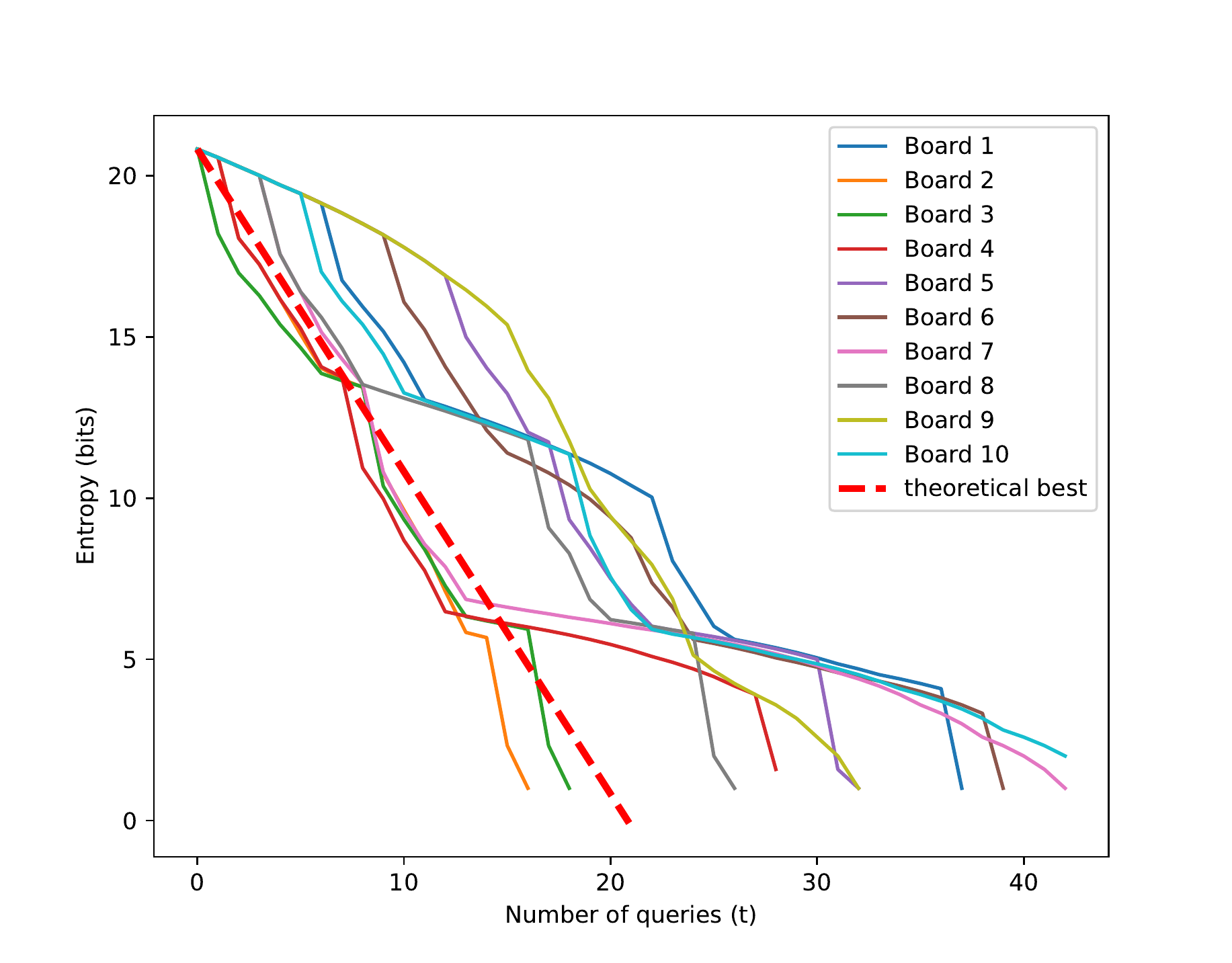}
    \caption{Algorithm against 10 random boards v.s. theoretical best average}
    \label{fig:Battleship10}
\end{figure}

We can observe two types of stages in the declination of entropy: at first, the entropy declination is usually slow, as the algorithm knows very little about the board and is exploring the board. When a "hit" occurs, the entropy drops rapidly, before the information provided by this hit is fully exploited, and this process repeats until entropy drops to 0.

For this small scale test, the results are very diverged. Some were very lucky and better than the theoretical best average $21$, and some are almost 2 times of theoretical best average. We ran the same test on a larger scale, on 1000 randomly chosen target boards and got some interesting results.

\begin{figure}[htb]
    \centering
    \includegraphics[width=0.5\textwidth]{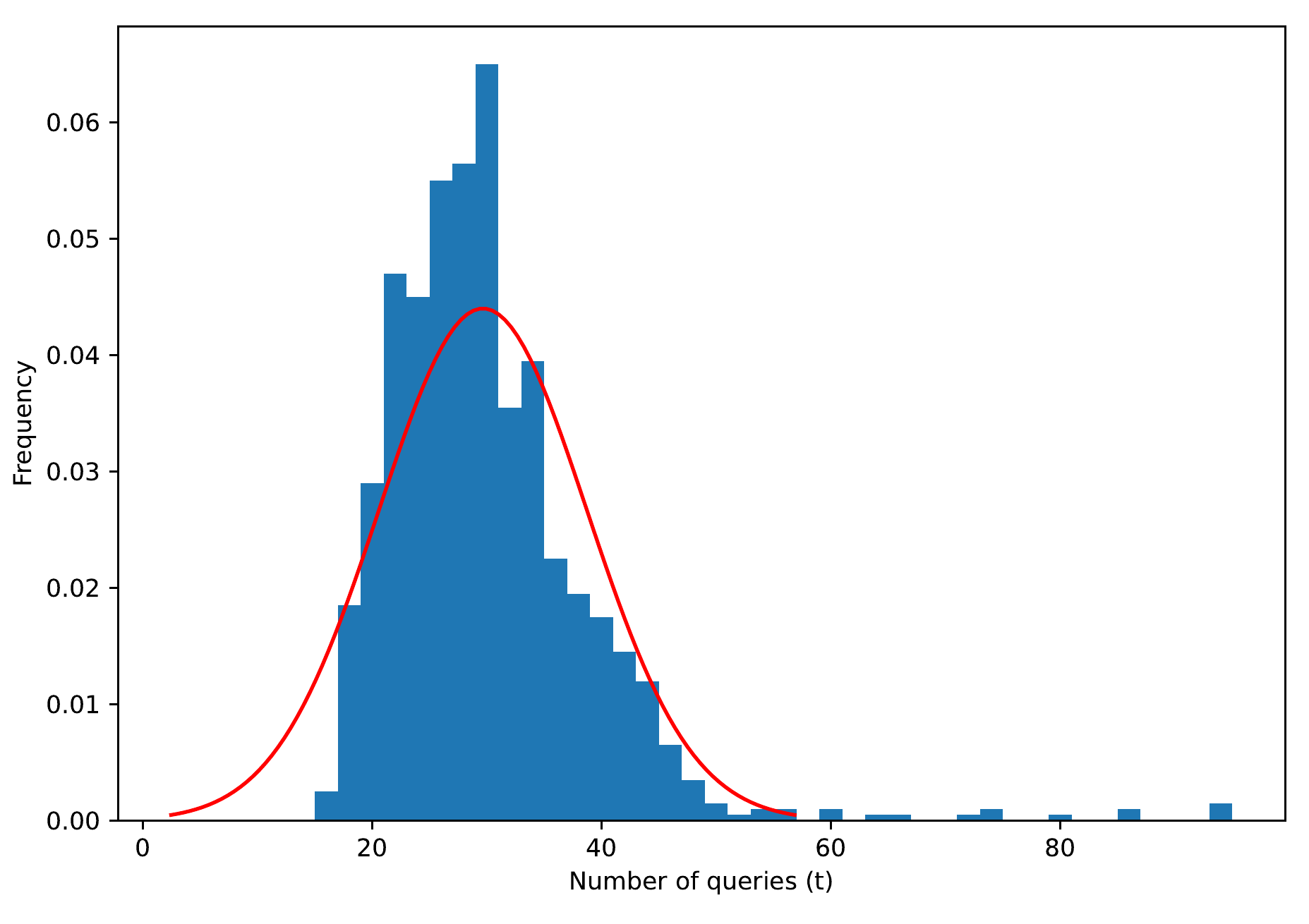}
    \caption{Distribution of queries to determine target board $X^*$}
    \label{fig:Hist1000}
\end{figure}
The mean number of queries $\bar{T}= 29.651$ and the standard deviation $\sigma =  9.062$. From this we can conclude that our GBSC-based algorithm performs reasonably well for Battleship, since $\bar{T}\approx 1.4 \bar{T}^*$. We can also take a look at the actual strategy of our algorithm for one game. This visualization (Fig.\ref{fig:Pmatrix}) shows how the algorithm prefers hit probability closest to 0.5.
\begin{figure}
    \centering
    \begin{subfigure}{0.35\textwidth}
        \includegraphics[scale=0.4]{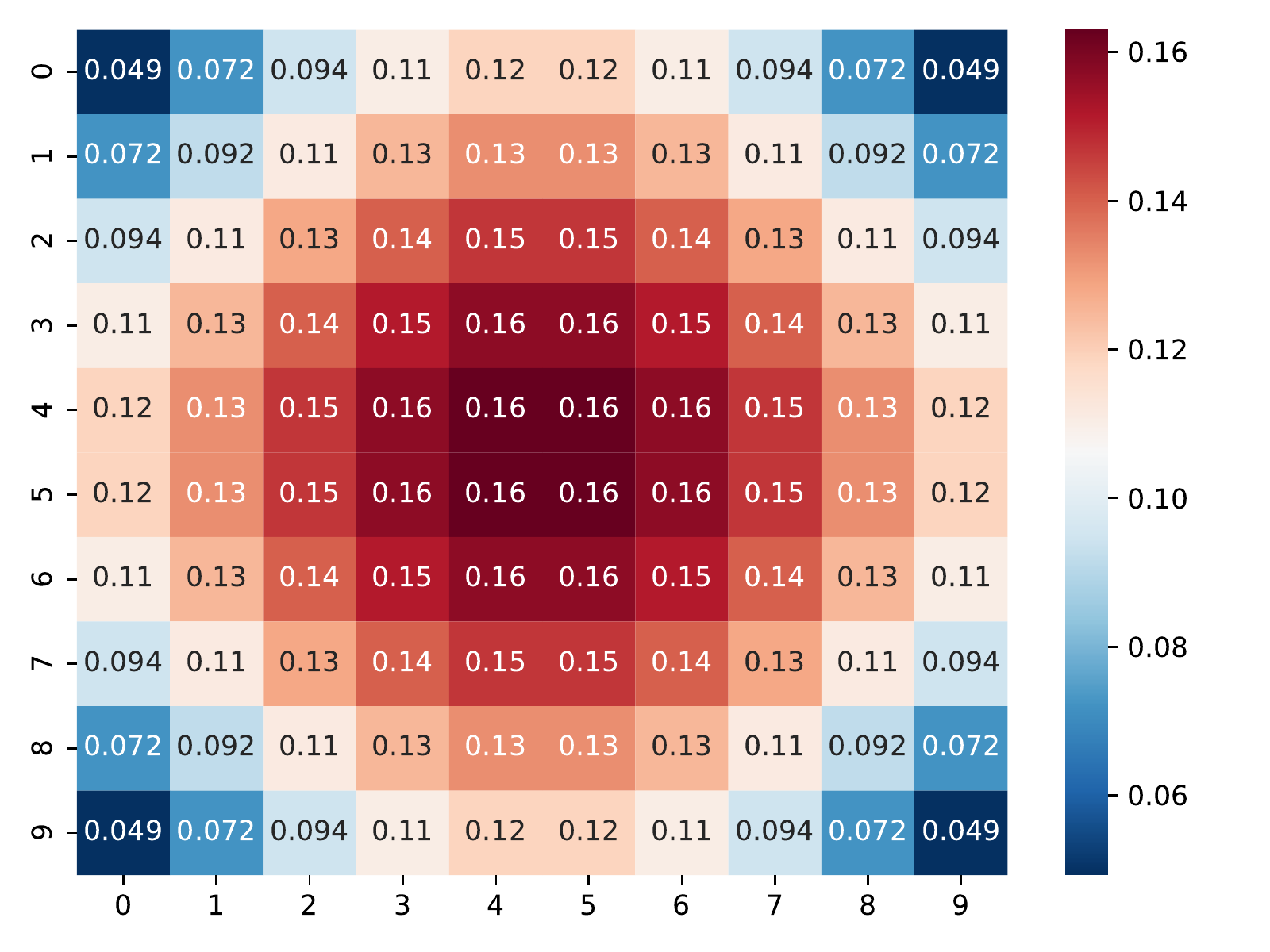}
        \caption{t=0}
    \end{subfigure}
    \begin{subfigure}{0.35\textwidth}
        \includegraphics[scale=0.4]{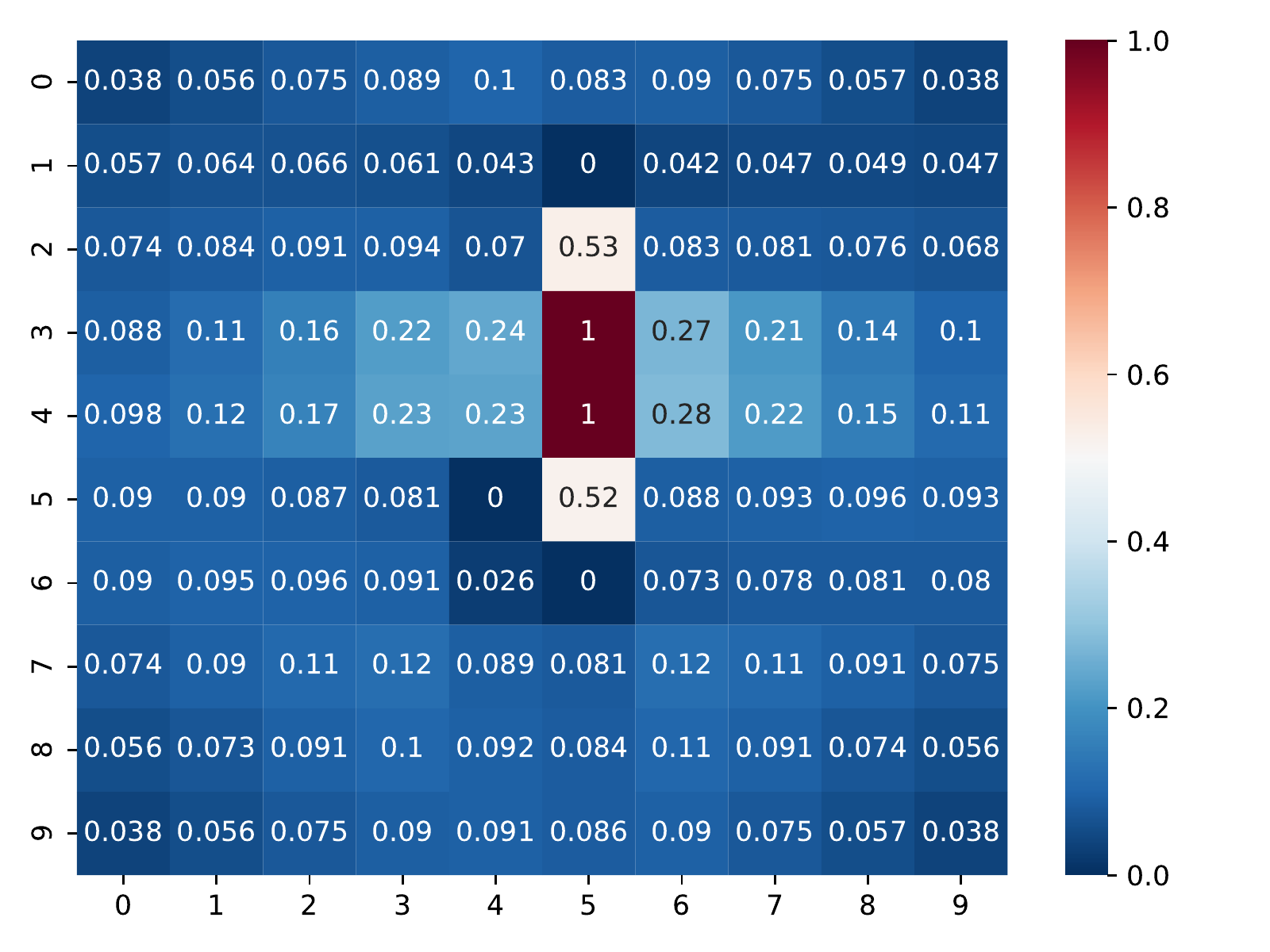}
        \caption{t=5}
    \end{subfigure}

    \begin{subfigure}{0.35\textwidth}
        \includegraphics[scale=0.4]{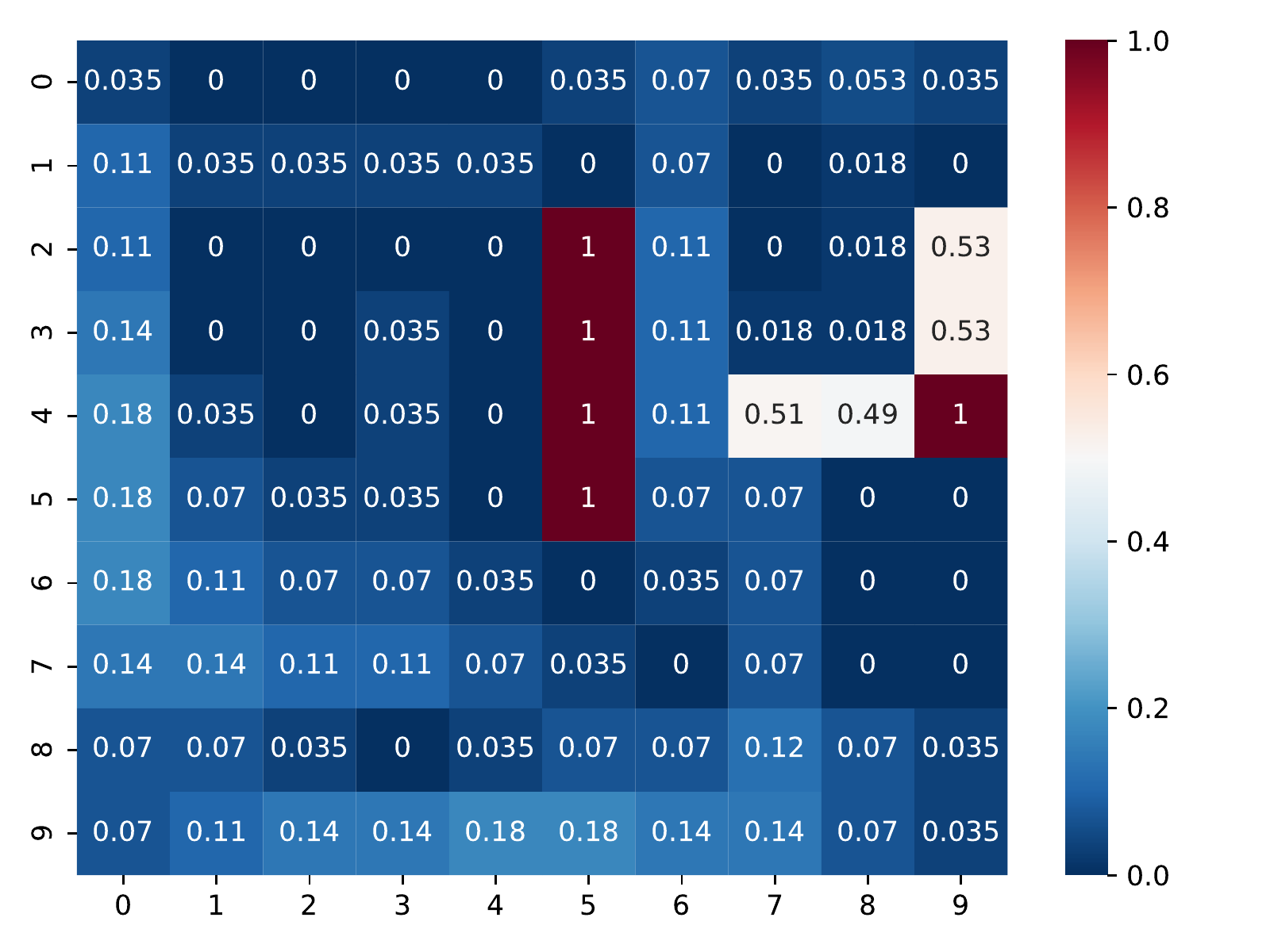}
        \caption{t=19}
    \end{subfigure}
    \begin{subfigure}{0.35\textwidth}
        \includegraphics[scale=0.4]{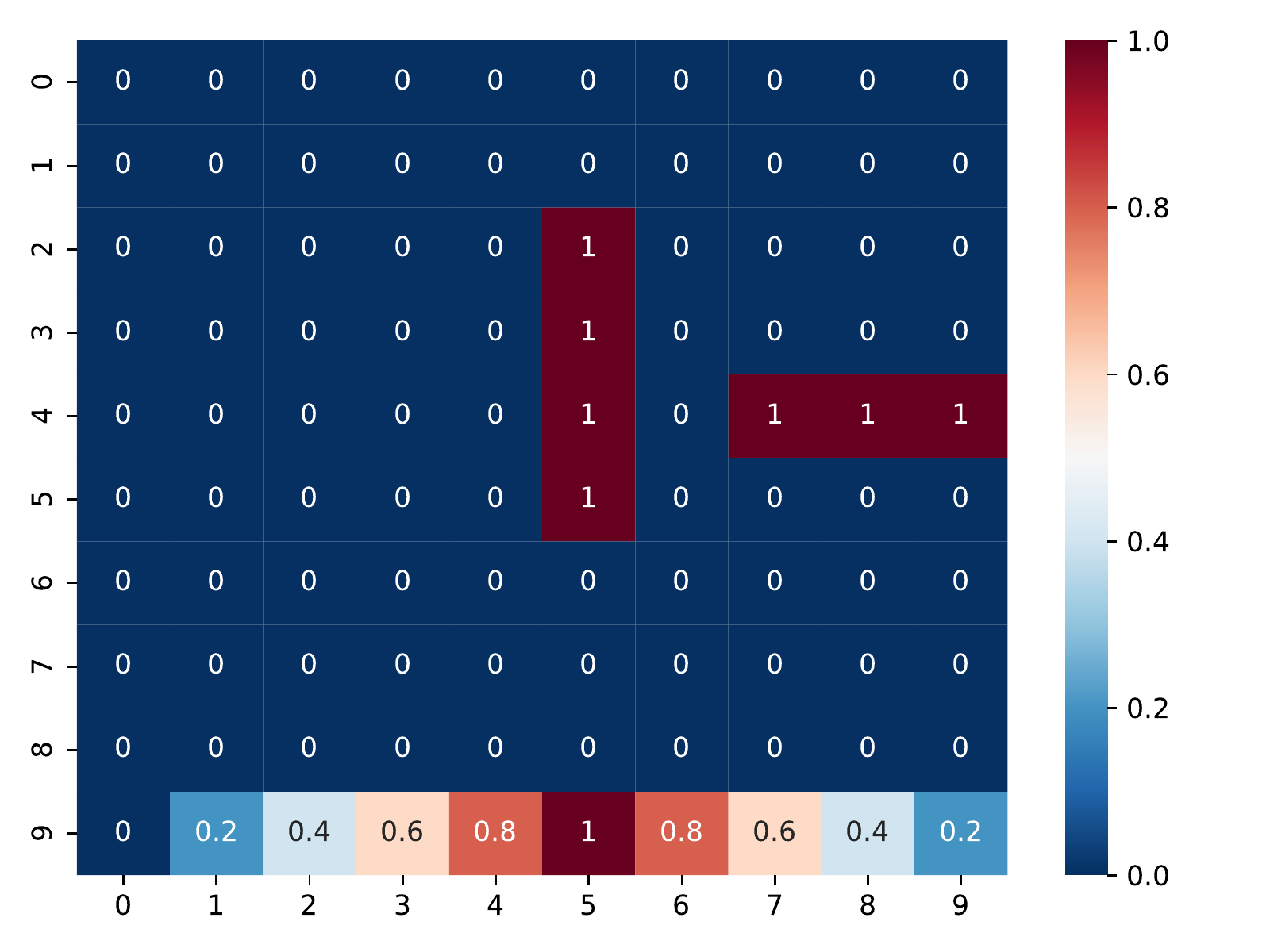}
        \caption{t=22}
    \end{subfigure}
    \caption{Hit probability matrix visualized. $p_{ij}:$ 0=No ship, 1=Has ship, (0,1)=Uncertain}
    \label{fig:Pmatrix}
\end{figure}

In another game, we demonstrate the hit pattern of our algorithm. (Fig.\ref{fig:Hitpattern}) Interestingly, our algorithm displays a diagonal search strategy, which is not intentionally designed in our GBSC based algorithm, but rather the result of computing the hit probability matrix, because a miss at $(i,j)$ reduces the conditional probability at $(i+1,j), (i-1,j),(i,j+1),(i,j-1)$. Also notice that the diagonal bombing lines are 3 grids apart, making a 4-long ship impossible to fit in the center area. Therefore, a diagonal searching pattern is great for quickly reducing entropy. This increases our confidence about this GBSC-based algorithm, since trying to hit the ships by bombing a diagonal pattern is a well-known strategy in this classic game, which also coincides with some previous deterministic approaches\cite{Rodin}.

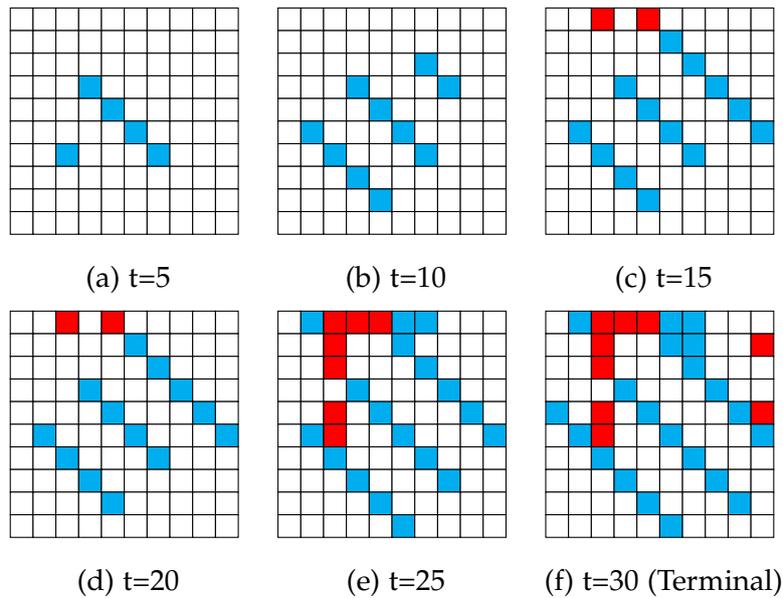
\begin{figure}[H]
    \begin{center}
    \begin{subfigure}{0.2\textwidth}
    \begin{tikzpicture}
    \matrix[boardstyle]{
        |[w]|& |[w]|& |[w]|& |[w]|& |[w]|& |[w]|& |[w]|& |[w]|& |[w]|& |[w]|\\
        |[w]|& |[w]|& |[w]|& |[w]|& |[w]|& |[w]|& |[w]|& |[w]|& |[w]|& |[w]|\\
        |[w]|& |[w]|& |[w]|& |[w]|& |[w]|& |[w]|& |[w]|& |[w]|& |[w]|& |[w]|\\
        |[w]|& |[w]|& |[w]|& |[b]|& |[w]|& |[w]|& |[w]|& |[w]|& |[w]|& |[w]|\\
        |[w]|& |[w]|& |[w]|& |[w]|& |[b]|& |[w]|& |[w]|& |[w]|& |[w]|& |[w]|\\
        |[w]|& |[w]|& |[w]|& |[w]|& |[w]|& |[b]|& |[w]|& |[w]|& |[w]|& |[w]|\\
        |[w]|& |[w]|& |[b]|& |[w]|& |[w]|& |[w]|& |[b]|& |[w]|& |[w]|& |[w]|\\
        |[w]|& |[w]|& |[w]|& |[w]|& |[w]|& |[w]|& |[w]|& |[w]|& |[w]|& |[w]|\\
        |[w]|& |[w]|& |[w]|& |[w]|& |[w]|& |[w]|& |[w]|& |[w]|& |[w]|& |[w]|\\
        |[w]|& |[w]|& |[w]|& |[w]|& |[w]|& |[w]|& |[w]|& |[w]|& |[w]|& |[w]|\\
        }; 
        \end{tikzpicture} 
        \caption{t=5}
        \end{subfigure}
    \begin{subfigure}{0.2\textwidth}
    \begin{tikzpicture}
    \matrix[boardstyle]{
        |[w]|& |[w]|& |[w]|& |[w]|& |[w]|& |[w]|& |[w]|& |[w]|& |[w]|& |[w]|\\
        |[w]|& |[w]|& |[w]|& |[w]|& |[w]|& |[w]|& |[w]|& |[w]|& |[w]|& |[w]|\\
        |[w]|& |[w]|& |[w]|& |[w]|& |[w]|& |[w]|& |[b]|& |[w]|& |[w]|& |[w]|\\
        |[w]|& |[w]|& |[w]|& |[b]|& |[w]|& |[w]|& |[w]|& |[b]|& |[w]|& |[w]|\\
        |[w]|& |[w]|& |[w]|& |[w]|& |[b]|& |[w]|& |[w]|& |[w]|& |[w]|& |[w]|\\
        |[w]|& |[b]|& |[w]|& |[w]|& |[w]|& |[b]|& |[w]|& |[w]|& |[w]|& |[w]|\\
        |[w]|& |[w]|& |[b]|& |[w]|& |[w]|& |[w]|& |[b]|& |[w]|& |[w]|& |[w]|\\
        |[w]|& |[w]|& |[w]|& |[b]|& |[w]|& |[w]|& |[w]|& |[w]|& |[w]|& |[w]|\\
        |[w]|& |[w]|& |[w]|& |[w]|& |[b]|& |[w]|& |[w]|& |[w]|& |[w]|& |[w]|\\
        |[w]|& |[w]|& |[w]|& |[w]|& |[w]|& |[w]|& |[w]|& |[w]|& |[w]|& |[w]|\\
        }; 
        \end{tikzpicture} 
        \caption{t=10}
    \end{subfigure}
    \begin{subfigure}{0.2\textwidth}
    \begin{tikzpicture}
    \matrix[boardstyle]{
        |[w]|& |[w]|& |[r]|& |[w]|& |[r]|& |[w]|& |[w]|& |[w]|& |[w]|& |[w]|\\
        |[w]|& |[w]|& |[w]|& |[w]|& |[w]|& |[b]|& |[w]|& |[w]|& |[w]|& |[w]|\\
        |[w]|& |[w]|& |[w]|& |[w]|& |[w]|& |[w]|& |[b]|& |[w]|& |[w]|& |[w]|\\
        |[w]|& |[w]|& |[w]|& |[b]|& |[w]|& |[w]|& |[w]|& |[b]|& |[w]|& |[w]|\\
        |[w]|& |[w]|& |[w]|& |[w]|& |[b]|& |[w]|& |[w]|& |[w]|& |[b]|& |[w]|\\
        |[w]|& |[b]|& |[w]|& |[w]|& |[w]|& |[b]|& |[w]|& |[w]|& |[w]|& |[b]|\\
        |[w]|& |[w]|& |[b]|& |[w]|& |[w]|& |[w]|& |[b]|& |[w]|& |[w]|& |[w]|\\
        |[w]|& |[w]|& |[w]|& |[b]|& |[w]|& |[w]|& |[w]|& |[w]|& |[w]|& |[w]|\\
        |[w]|& |[w]|& |[w]|& |[w]|& |[b]|& |[w]|& |[w]|& |[w]|& |[w]|& |[w]|\\
        |[w]|& |[w]|& |[w]|& |[w]|& |[w]|& |[w]|& |[w]|& |[w]|& |[w]|& |[w]|\\
        };
        \end{tikzpicture} 
        \caption{t=15}
    \end{subfigure}

    \begin{subfigure}{0.2\textwidth}
    \begin{tikzpicture}
    \matrix[boardstyle]{
        |[w]|& |[w]|& |[r]|& |[w]|& |[r]|& |[w]|& |[w]|& |[w]|& |[w]|& |[w]|\\
        |[w]|& |[w]|& |[w]|& |[w]|& |[w]|& |[b]|& |[w]|& |[w]|& |[w]|& |[w]|\\
        |[w]|& |[w]|& |[w]|& |[w]|& |[w]|& |[w]|& |[b]|& |[w]|& |[w]|& |[w]|\\
        |[w]|& |[w]|& |[w]|& |[b]|& |[w]|& |[w]|& |[w]|& |[b]|& |[w]|& |[w]|\\
        |[w]|& |[w]|& |[w]|& |[w]|& |[b]|& |[w]|& |[w]|& |[w]|& |[b]|& |[w]|\\
        |[w]|& |[b]|& |[w]|& |[w]|& |[w]|& |[b]|& |[w]|& |[w]|& |[w]|& |[b]|\\
        |[w]|& |[w]|& |[b]|& |[w]|& |[w]|& |[w]|& |[b]|& |[w]|& |[w]|& |[w]|\\
        |[w]|& |[w]|& |[w]|& |[b]|& |[w]|& |[w]|& |[w]|& |[w]|& |[w]|& |[w]|\\
        |[w]|& |[w]|& |[w]|& |[w]|& |[b]|& |[w]|& |[w]|& |[w]|& |[w]|& |[w]|\\
        |[w]|& |[w]|& |[w]|& |[w]|& |[w]|& |[w]|& |[w]|& |[w]|& |[w]|& |[w]|\\
        };
        \end{tikzpicture} 
        \caption{t=20}
    \end{subfigure}
    \begin{subfigure}{0.2\textwidth}
    \begin{tikzpicture}
    \matrix[boardstyle]{
        |[w]|& |[b]|& |[r]|& |[r]|& |[r]|& |[b]|& |[b]|& |[w]|& |[w]|& |[w]|\\
        |[w]|& |[w]|& |[r]|& |[w]|& |[w]|& |[b]|& |[w]|& |[w]|& |[w]|& |[w]|\\
        |[w]|& |[w]|& |[r]|& |[w]|& |[w]|& |[w]|& |[b]|& |[w]|& |[w]|& |[w]|\\
        |[w]|& |[w]|& |[w]|& |[b]|& |[w]|& |[w]|& |[w]|& |[b]|& |[w]|& |[w]|\\
        |[w]|& |[w]|& |[r]|& |[w]|& |[b]|& |[w]|& |[w]|& |[w]|& |[b]|& |[w]|\\
        |[w]|& |[b]|& |[r]|& |[w]|& |[w]|& |[b]|& |[w]|& |[w]|& |[w]|& |[b]|\\
        |[w]|& |[w]|& |[b]|& |[w]|& |[w]|& |[w]|& |[b]|& |[w]|& |[w]|& |[w]|\\
        |[w]|& |[w]|& |[w]|& |[b]|& |[w]|& |[w]|& |[w]|& |[b]|& |[w]|& |[w]|\\
        |[w]|& |[w]|& |[w]|& |[w]|& |[b]|& |[w]|& |[w]|& |[w]|& |[w]|& |[w]|\\
        |[w]|& |[w]|& |[w]|& |[w]|& |[w]|& |[b]|& |[w]|& |[w]|& |[w]|& |[w]|\\
        };
        \end{tikzpicture} 
        \caption{t=25}
    \end{subfigure}
     \begin{subfigure}{0.2\textwidth}
    \begin{tikzpicture}
    \matrix[boardstyle]{
        |[w]|& |[b]|& |[r]|& |[r]|& |[r]|& |[b]|& |[b]|& |[w]|& |[w]|& |[w]|\\
        |[w]|& |[w]|& |[r]|& |[w]|& |[w]|& |[b]|& |[b]|& |[w]|& |[w]|& |[r]|\\
        |[w]|& |[w]|& |[r]|& |[w]|& |[w]|& |[w]|& |[b]|& |[w]|& |[w]|& |[w]|\\
        |[w]|& |[w]|& |[w]|& |[b]|& |[w]|& |[w]|& |[w]|& |[b]|& |[w]|& |[w]|\\
        |[b]|& |[w]|& |[r]|& |[w]|& |[b]|& |[w]|& |[w]|& |[w]|& |[b]|& |[r]|\\
        |[w]|& |[b]|& |[r]|& |[w]|& |[w]|& |[b]|& |[w]|& |[w]|& |[w]|& |[b]|\\
        |[w]|& |[w]|& |[b]|& |[w]|& |[w]|& |[w]|& |[b]|& |[w]|& |[w]|& |[w]|\\
        |[w]|& |[w]|& |[w]|& |[b]|& |[w]|& |[w]|& |[w]|& |[b]|& |[w]|& |[w]|\\
        |[w]|& |[w]|& |[w]|& |[w]|& |[b]|& |[w]|& |[w]|& |[w]|& |[b]|& |[w]|\\
        |[w]|& |[w]|& |[w]|& |[w]|& |[w]|& |[b]|& |[w]|& |[w]|& |[w]|& |[w]|\\
        };
        \end{tikzpicture} 
        \caption{t=30 (Terminal)}
    \end{subfigure}
    \caption{Hit pattern for a typical Battleship game played by GBSC-based algorithm}
    \label{fig:Hitpattern}
    \end{center}
\end{figure}

\subsection{Another Try to Solving the DNA Detection Problem}
Besides the greedy algorithm based on huffman coding, we can also use the GBSC to solve the DNA detection problem.

We compared the output and time cost of the two algorithms by using the same random seed to generate random sequences (assuming probability at every point is iid) of different lengths and using the two algorithms to compute the expected number of detections respectively. We generated 100 sequences for each length to reduce random error. Here are the results.

The average time cost for two algorithm is shown in Fig.\ref{fig:runningtime}.
\begin{figure}[H]
    \centering
    \includegraphics[width=0.45\textwidth]{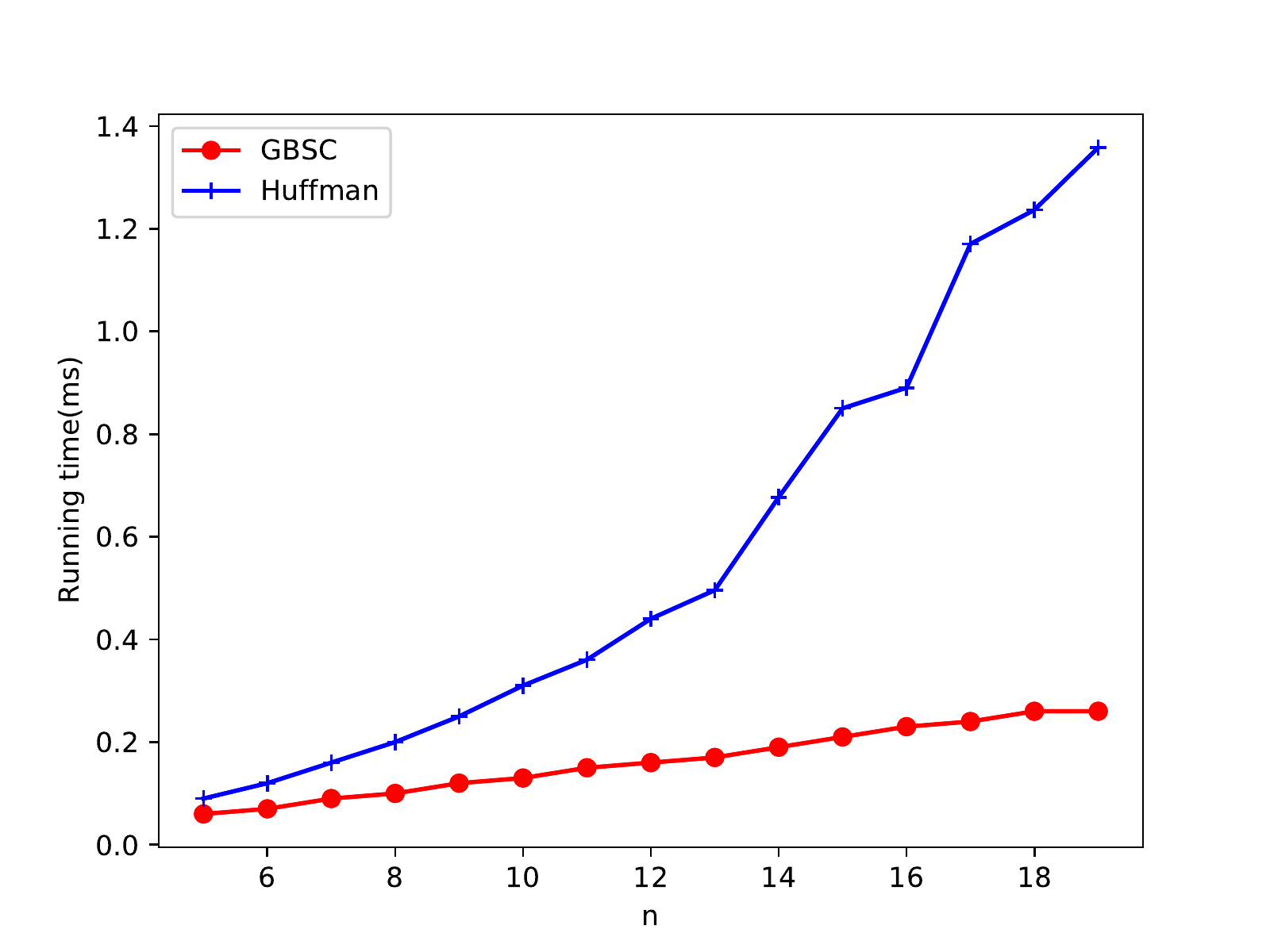}
    \caption{Running time of the two algorithms.}
    \label{fig:runningtime}
\end{figure}

It can be seen from the figures that the time complexity of GBSC is nearly linear, while the time complexity of the Huffman-based algorithm is not. The GBSC-based algorithm runs much faster than the Huffman-based algorithm algorithm. Besides, the the Huffman-based algorithm is not robust, and would take quite a long time to give an output when initial data is not very ideal. For example, one of the randomly generated distributions with size 36 took it took about 30 seconds for the Huffman-based algorithm to calculate the result, while most other distributions usually took about 10ms. In comparison, the GBSC reliably gave results in a relatively short time.

The expected lengths calculated by two algorithm are shown in Fig.\ref{fig:GBSCvsHuffman}.
\begin{figure}[H]
    \centering
    \includegraphics[width=0.45\textwidth]{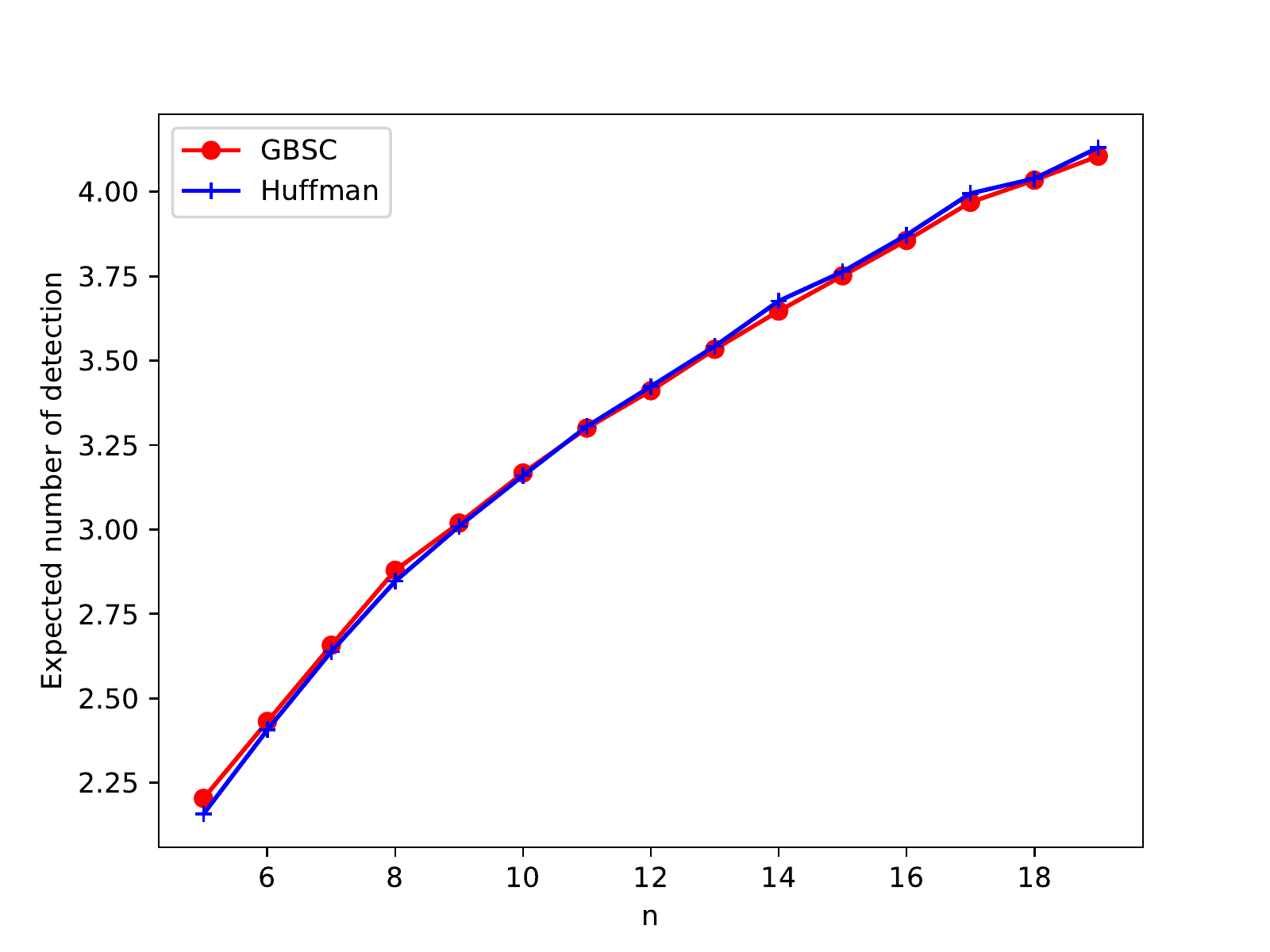}
    \caption{Expected number of detections calculated by the two algorithms.}
     \label{fig:GBSCvsHuffman}
\end{figure}
It is easy to see that the two figures are very similar, which means that the outputs of two algorithm in our experiment are almost the same.


\section{Future Work}

Here we propose some aspects that can be further researched.

\begin{itemize}
    \item For the DNA detection problem, our algorithm only works the simplified problem where there is only one target exon. A more general efficient multi-exon detection algorithm remains to be discovered.

    \item Finding the true optimal solution of the DNA detection problem and the Battleship problem is probably a NP-hard problem. Can we prove it?

    \item In this report, we implicitly assume that the distribution $p(x)$ is independent, i.e., $\forall i, j \in \mathcal{X}$, $p(i)$ and $p(j)$ are independent. If this condition is canceled, the shape of possible decision trees will not be arbitrary, making a better algorithm possible.
    
    \item Currently the our GBSC algorithm runs brute force to find the separation with probability closest to $\frac{1}{2}$. This can be time-costly for larger problems, for example, Battleship with a larger board or more ships, so there may be room for optimizations, for example maybe use Monte-Carlo method to get a good-enough approximation.
\end{itemize}


\section{Conclusion}

In this paper, we purposed a generalized model to similar querying problems where the questions we ask are limited to a subset $\mathscr{A} \subsetneq 2^\mathcal{X}$, by analyzing three different types of problems: Huffman-coding-equivalent problems, the 1-player Battleship problem, and the DNA detection problem. 

Inspired by Huffman coding and greedy decision trees, we proposed two coding schemes: one based on Huffman coding that merges two nodes greedily if possible, and another called GBSC (Greedy Binary Separation Coding). We then proved that the expected code length of GBSC achieves the Shannon Code bound, meaning GBSC is at least as good as Shannon Coding. 

To see the effectiveness of these two algorithms in real world applications, we applied the Huffman-based coding to the DNA detection problem, and GBSC to both the 1-player Battleship and the DNA detection problem to see how well they perform in terms of average coding length, or equivalently, the average number of queries to determine the target random variable. For the DNA detection problem, out of \num{10000} tests, the Huffman-based algorithm gives expected length $L$ less than $1.1L^*$ in most cases, where $L^*$ is the true optimal calculated by brute force. We also compared the Huffman-based algorithm and GBSC algorithm, and found that while the two algorithm yielded very similar expected lengths, GBSC was much more time-efficient and therefore suitable for larger problems. In 1-player Battleship with a $10\times 10$ board and 3 ships, out of the \num{1000} tests, the GBSC-based algorithm achieves an average of $\bar{T}\approx 1.4 \bar{T}^*=29.651$ queries, where $\bar{T}^*$ is the theoretical best average number of queries. Overall, we can conclude that GBSC and the Huffman-based greedy code performs decently well in solving query-based decision problems that cannot be solved by the original Huffman coding algorithm.
\newpage 

\pagenumbering{roman} 

\bibliography{literature.bib} 
\bibliographystyle{IEEEtran} 

\newpage 
\section*{Appendix} 
\pagenumbering{Roman} 

The code of the experiments in the paper is available at  \url{https://github.com/MadCreeper/Constrained-Optimal-Querying-Huffman-Coding-and-Beyond}.



\end{document}